\documentclass[11pt,a4paper]{article}

\usepackage{epsf,epsfig,amsfonts,amsgen,amsmath,amstext,amsbsy,amsopn,amsthm,cases,listings,color}
\usepackage{ebezier,eepic}
\usepackage{color}
\usepackage{multirow}
\usepackage{epstopdf}
\usepackage{graphicx}   
\usepackage{pgf,tikz}
\usepackage{mathrsfs}
\usepackage[marginal]{footmisc}
\usepackage{enumerate}
\usepackage{enumitem}
\usepackage[titletoc]{appendix}
\usepackage{booktabs}
\usepackage{url}
\usepackage{mathtools}
\usepackage{pdflscape} 
\usepackage{pgfplots}
\usepackage{authblk}
\usepackage{amssymb}
\usepackage{wasysym}
\usepackage{empheq}
\usepackage{dsfont}
\usepackage{tikz}
\usepackage{longtable}
\usepackage{float}
\usepackage{subfigure}
\pgfplotsset{compat=1.18}
\usepackage{mathrsfs}
\usepackage{wasysym} 
\usetikzlibrary{arrows}
\usepackage{aligned-overset}
\usepackage{bm}
\usepackage{bbm}
\usepackage[T1]{fontenc}
\usepackage{amsmath}
\usepackage[ruled,linesnumbered,vlined]{algorithm2e}
\usepackage{makecell}

\allowdisplaybreaks[1]

\theoremstyle{definition}
\newtheorem{definition}{Definition} [section]
\newtheorem{theorem}[definition]{Theorem}
\newtheorem{lemma}[definition]{Lemma}

\begin{document}
\title{\bf\Large  Near-Optimal Algorithms for Maximal Clique Enumeration in Structurally Sparse Graphs}
\author{
Jianfeng Hou\thanks{
Emails: jfhou@fzu.edu.cn},
Hongbin Zhao\thanks{
Emails: hbzhao2024@163.com },
\\
\small Center for Discrete Mathematics and Theoretical Computer Science,\\
\small Fuzhou University, Fuzhou, Fujian, China
}
\date{}
\maketitle
\begin{abstract}
We study the exact enumeration of maximal cliques in graph classes defined by excluded  clique minors and excluded  clique immersions. 
For $n$-vertex $K_t$-minor-free graphs, we give an algorithm that lists all maximal cliques in $n \cdot 4^{2t/5+o(t)}$ time, significantly improving the previous $n \cdot 2^{O(t \log \log t)}$ bound of Eppstein, L\"{o}ffler, and Strash. 
For $n$-vertex $K_t$-immersion-free graphs, we establish the first exact enumeration algorithm parameterized by immersion number, achieving a running time of $n \cdot 3^{t/3+o(t)}$. 
While both algorithms employ a common degeneracy-based root-assignment scheme, their analyses require distinct structural mechanisms. 
Crucially, rather than applying generic sparsity bounds, our algorithms deeply integrate the specific structural obstructions—local density thresholds for minors and minimum-degree branchings for immersions—directly into the enumeration logic. 
We also prove matching output-size lower bounds, up to sub-exponential factors in $t$, using specialized constructions.   Consequently, the exponential bases $4^{2/5}$ and $3^{1/3}$ are asymptotically optimal.

\medskip
\end{abstract}

\section{Introduction}\label{sec:int}

Let $G$ be a graph.  A \emph{clique} of $G$ is a set of pairwise adjacent vertices, and it is called a \emph{maximal clique} if it is not properly contained in another clique. Maximal clique enumeration is a canonical output-sensitive problem in graph algorithms, with applications spanning community detection \cite{WCLG17}, web network analysis \cite{STKA07,WHYW09}, bioinformatics \cite{M22,ZPKF08}, telecommunications \cite{BB06}, and location privacy protection \cite{PM13}. The proliferation of massive real-world graphs has made the design of efficient enumeration algorithms an increasingly critical challenge. 
From the viewpoint of worst-case complexity, however, the problem is necessarily exponential: the maximum clique problem is NP-hard \cite{K72}, and an $n$-vertex graph may contain $3^{n/3}$ maximal cliques by the classical result of Miller--Muller \cite{MM60} and independently Moon--Moser \cite{MM65}. 
This bound is algorithmically tight in general graphs, as Tomita, Tanaka, and Takahashi \cite{TTT04} gave an optimized Bron--Kerbosch algorithm \cite{BK73} with worst-case running time $O(3^{n/3})$.

The central algorithmic question is therefore not whether the exponential dependence can be avoided in full generality, but whether it can be confined to a structural parameter. This perspective is particularly successful for sparse or locally sparse graph classes. 
For example, maximal clique enumeration admits fixed-parameter algorithms in terms of degeneracy \cite{ELS10,ELS13} and in the model of $c$-closed graphs \cite{FRSWW18,FRSWW20}; related centralized \cite{BK73,CK08,ELS13,K01,TTT04,JXHZZ22} and distributed \cite{CKFYZ11,CZKC12,CDM16,XCF16} algorithms have also been developed for large sparse networks. 
In this paper, we focus on two fundamental containment-based notions of sparsity: excluding a complete graph as a minor and as an immersion.

We recall the standard definition of graph minors. A graph $H$ is a \emph{minor} of a graph $G$ if $H$ can be obtained from $G$ by a sequence of vertex deletions, edge deletions, and edge contractions. For an integer $t\ge 1$, let $K_t$ denote the complete graph on $t$ vertices. A graph $G$ is \emph{$K_t$-minor-free} if $K_t$ is not a minor of $G$. A fundamental result of Thomason \cite{T01} shows that every $K_t$-minor-free graph has bounded degeneracy. Using the degeneracy-based algorithmic framework,  Eppstein, L\"{o}ffler, and Strash \cite{ELS13}  established the following parameterized enumeration bound.
\begin{theorem}[Eppstein--L\"{o}ffler--Strash \cite{ELS13}] \label{Thm:ELS}
    For positive integers $n, t$ and  $n$-vertex $K_t$-minor free graph $G$, there exists an algorithm that enumerates all maximal cliques of $G$ in $n \cdot 2^{O(t \log \log t)}$ time.
\end{theorem}
Theorem \ref{Thm:ELS} provides a fixed-parameter tractable (FPT) algorithm, but it does not identify the correct exponential dependence on $t$.  Our first result determines the asymptotically optimal exponential term for the minor-free case.
\begin{theorem}\label{Thm:Alg}
    For every $n$-vertex $K_t$-minor free graph, all its maximal cliques can be listed in $n \cdot 4^{2t/5+o(t)}$ time. Furthermore, this bound is asymptotically sharp up to the $o(t)$ term in the exponent when $n\ge 8t/5$.
\end{theorem}

We also consider the analogous problem for forbidden immersions. An \emph{immersion} of a graph $H$ into $G$ maps the vertices of $H$ injectively to branch vertices of $G$, while each edge of $H$ is represented by an edge-disjoint path between the corresponding pair of branch vertices. In this work, we adopt the standard strong version of immersion, where no branch vertex appears as an internal vertex of one of these paths. Graphs that exclude a fixed immersion admit rich structural properties, which yield significant implications for routing problems and network classification tasks \cite{AL03,RS10,W15}. Immersions impose a different type of structural restriction from minors: minor containment is driven by contractions, whereas immersion containment is driven by edge-disjoint routing. This distinction leads to a different tight exponential base.
\begin{theorem}\label{Thm:immersion}
    For every $n$-vertex graph with no $K_t$-immersion, all its maximal cliques can be listed in $n \cdot 3^{t/3+o(t)}$ time. Furthermore, this bound is asymptotically sharp up to the $o(t)$ term in the exponent when $n\ge t$.
\end{theorem}

The resulting parameterized enumeration bounds for both structurally sparse graph classes are summarized in Table \ref{tab:results_comparison}.
\begin{table}[htpb]

    \centering
    \label{tab:results_comparison}
    \begin{tabular}{@{}lllc@{}}
        \toprule
        \textbf{Graph Class} & \textbf{Previous Best Bound} & \textbf{Our Bound} & \makecell[c]{\textbf{Optimal exponential}\\ \textbf{term?}} \\
        \midrule
        $K_t$-minor-free & $n \cdot 2^{O(t \log \log t)}$ \cite{ELS13} & $n \cdot 4^{2t/5 + o(t)}$ & \textbf{Yes} \\
        $K_t$-immersion-free & --- (Open problem) & $n \cdot 3^{t/3 + o(t)}$ & \textbf{Yes} \\
        \bottomrule
    \end{tabular}
    \caption{Maximal clique enumeration bounds on structurally sparse graphs.}
\end{table}

\textbf{Algorithmic Ideas.} 
Both algorithms assign every maximal clique to its first vertex in a degeneracy ordering. After fixing this root, the remaining search is carried out only inside its later neighborhood. The recursive state consists of a current clique $C$, a candidate set $P$, and a set $X$ of previously excluded later vertices that must be checked for maximality. This root assignment keeps the dependence on the number of vertices linear after summing over all roots. The sub-exponential error term $o(t)$ originates purely from the polylogarithmic depth limits of our localized recursion trees.

\begin{itemize}
    \item \textbf{The minor-free case}: The recursion is not analyzed only through degeneracy. Whenever the candidate graph becomes small we finish using the standard worst-case maximal-clique algorithm \cite{TTT04}, and whenever it becomes sufficiently dense we switch to an output-sensitive enumeration routine \cite{TIAS77}. The correctness of the dense stopping rule uses the relation between local density and clique minors due to Fox and Wei \cite{FW17}, together with a tight bound on the number of maximal cliques in dense terminal graphs. The remaining recursion is controlled by a peeling argument: each inclusion step removes a non-negligible number of candidates, so the number of terminal subproblems is only sub-exponential in $t$.
    \item \textbf{The immersion-free case}: The same root-assignment framework applies, but the structural reason for termination is different. A structural criterion of Fox and Wei \cite{FW20} implies that, above the threshold $|P| > t$, a minimum-degree branching vertex reduces the excess $|P| - t$ by a constant factor along inclusion branches. This potential-function analysis yields only $2^{O(\log^2 t)}$ terminal subproblems per root, and the candidate set in each terminal subproblem has at most $t$ vertices. The resulting bound is therefore $n \cdot 3^{t/3+o(t)}$.
\end{itemize}

Finally, the upper bounds are tight in the natural sense for enumeration algorithms: there are graphs in the corresponding classes with that many maximal cliques. Complete multipartite graphs with parts of size four give $n \cdot 4^{2t/5-o(t)}$ maximal cliques while excluding a $K_t$-minor, and complete multipartite graphs with parts of size three give $n \cdot 3^{t/3-o(t)}$ maximal cliques while excluding a $K_t$-immersion. Hence the exponential bases $4^{2t/5}$ and $3^{t/3}$ cannot be improved, up to sub-exponential factors in $t$.

\textbf{Organization.} Section \ref{sec:preliminaries} provides the prerequisite structural tools along with the dense-graph bound for maximal cliques required for terminal cases. Section \ref{sec:minor_alg} presents the algorithm for $K_t$-minor-free graphs and proves Theorem \ref{Thm:Alg}. Finally, Section \ref{sec:immersion_alg} gives the algorithm for $K_t$-immersion-free graphs and proves Theorem \ref{Thm:immersion}.
\section{Preliminaries}\label{sec:preliminaries}
This section provides the fundamental structural tools required to prove our main theorems. The core philosophy of our analysis draws upon the peeling process developed by Fox--Wei \cite{FW17,FW20} and Shi--Wei \cite{SW25} to systematically isolate dense subgraphs. This approach conceptually echoes the central principles of the hypergraph container method developed independently by Balogh--Morris--Samotij \cite{BMS15} and Saxton--Thomason \cite{ST15}. Specifically, this strategy shifts the enumeration burden from general graphs to locally dense settings. To rigorously define what constitutes a ``dense'' terminal graph in our minor-free context, we rely on the following result by Fox and Wei \cite{FW17}. For a graph $G$, let $h(G)$ denote the \emph{Hadwiger number} of  $G$, defined as the maximum order of a clique minor in $G$.
\begin{lemma}[Fox--Wei \cite{FW17}] \label{Lem:FW}
Let $G$ be a graph on $n$ vertices with clique number $\omega$, and let $\Delta $ be the maximum degree of the complement of $G$.
If $n \geq \omega + 2\Delta^2 + 2$, then $h(G) = \left\lfloor \frac{n+\omega}{2} \right\rfloor$.
\end{lemma}
For the immersion setting, we rely on a result by Fox--Wei \cite{FW20}, which establishes that a minimum degree condition within any $t$-sized subset forces a strong $K_t$-immersion.
\begin{lemma}[Fox--Wei \cite{FW20}]\label{Lem:FW20}
    If a graph $G$ on $n$ vertices contains a vertex subset $T$ of size $t$ where every vertex in $T$ has a degree strictly greater than $\frac{n+t}{2}-2$, then $G$ contains a strong $K_t$-immersion with $T$ as the set of end vertices.
\end{lemma}
Finally, to analyze the termination conditions of our enumeration framework, we establish a tight combinatorial upper bound on the number of maximal cliques within highly dense subgraphs. For any graph $G$, let $v(G)$ denote the number of vertices, $\omega(G)$ the clique number, and $\mathrm{mc}(G)$ the total number of maximal cliques. 
\begin{lemma} \label{Lem:vtxCli}
    If $v(G)+\omega(G)\le s$, then $\mathrm{mc}(G) \le 4^{s/5}$.
\end{lemma}

\begin{proof}[Proof of Lemma \ref{Lem:vtxCli}]

Let $G$ be a graph maximizing the number of maximal cliques $\mathrm{mc}(G)$ subject to
the constraint $v(G) + \omega(G) \le s$. By Zykov symmetrization \cite{Zyk49}, we may
assume  that $G$ is a complete $k$-partite graph for some
$k \ge 2$. To see this, for each vertex $x \in V(G)$, let $c(x)$ denote the number of
maximal cliques of $G$ containing $x$. Given any non-adjacent pair $u, v \in V(G)$. Without loss of generality, let $c(u) \ge c(v)$. We construct a new graph $G'$ by replacing $v$ with a clone $u'$ of
$u$, that is, a new vertex $u'$ that is non-adjacent to $u$ and shares the same
neighborhood as $u$. Since no larger cliques are introduced,  the constraint $v(G') + \omega(G') \le s$ is preserved. Note that removing $v$ destroys at most $c(v)$ maximal cliques, while adding $u'$ creates exactly $c(u)$ new maximal cliques. Therefore,
\begin{equation*}
    \mathrm{mc}(G') \ge \mathrm{mc}(G) - c(v) + c(u) \ge \mathrm{mc}(G).
\end{equation*}

Let $G$ be a complete $k$-partite graph with part sizes $a_1, a_2, \dots, a_k \ge 1$. Then 
    \[
     v(G) = \sum_{i=1}^k a_i, \quad \omega(G) = k, \quad {\text and }\quad \text{mc}(G) = \prod_{i=1}^k a_i. 
    \]
It follows from  $v(G) + \omega(G) \le s$ that $\sum_{i=1}^k (a_i + 1) \le s$.
    To bound the product $\prod_{i=1}^k a_i$, we decouple the variables by seeking a tight constant base $C > 1$ such that $a_i \le C^{a_i+1}$ for all integers $a_i \ge 1$. This yields
    \begin{align}\label{ine:uppMC}
        \mathrm{mc}(G) = \prod_{i=1}^k a_i \le \prod_{i=1}^k C^{a_i + 1} = C^{\sum_{i=1}^k (a_i + 1)} \le C^s.
    \end{align}

To complete the proof, it suffices to find the smallest $C > 1$ satisfying
$C \ge a^{1/(a+1)}$ for all positive integers $a$. Consider the function
$f(x) = x^{1/(x+1)}$ for $x \ge 1$. A straightforward calculus argument shows that
there exists $x_0 \in (3, 4)$ such that $f$ is strictly increasing on $[1, x_0]$ and
strictly decreasing on $[x_0, +\infty)$. Since $x_0$ lies between $3$ and $4$, the
maximum of $f$ over positive integers is attained at either $a = 3$ or $a = 4$.
Comparing the two values,
$f(3) = 3^{1/4} \approx 1.3161 < 4^{1/5} \approx 1.3195 = f(4),$
so $f$ attains its maximum over positive integers at $a = 4$, giving
$C = 4^{1/5}$. Substituting into \eqref{ine:uppMC}, we obtain
\begin{equation*}
    \mathrm{mc}(G) \le 4^{s/5}. \qedhere
\end{equation*}
\end{proof}


%
\section{Maximal cliques in $K_t$-minor-free graphs}\label{sec:minor_alg}
In this section, we prove Theorem \ref{Thm:Alg}. Before that, we  give the revised enumeration framework. The main technical point is that the recursion is not run directly on the whole graph with a large global exclusion set. Instead, we first assign every maximal clique to its first vertex in a fixed degeneracy order. This removes the possible quadratic overhead caused by repeatedly scanning large exclusion sets.
To efficiently terminate the search and prune the recursion tree, we integrate two standard exact enumeration subroutines:
\begin{itemize}
    \item \textbf{The Tomita--Tanaka--Takahashi Algorithm \cite{TTT04}:} This optimized variant of the Bron--Kerbosch algorithm \cite{BK73} utilizes a pivot selection strategy to rigorously bound the worst-case time complexity. It guarantees the generation of all maximal cliques in an $n$-vertex graph in $O(3^{n/3})$ time. We deploy this subroutine to handle small remaining subgraphs (Step 3). We denote a call to this subroutine on a graph $H$ as $\text{\textbf{OptimizedBK}}(H)$.
    
    \item \textbf{The Tsukiyama--Ide--Ariyoshi--Shirakawa Algorithm \cite{TIAS77}:} This is a seminal output-sensitive enumeration algorithm. It generates all maximal cliques in a graph $G$ in time $O(|V||E| \cdot \mathrm{mc}(G))$.  We apply this algorithm to highly dense terminal graphs (Step 4). 
    We denote a call to this subroutine on a graph $H$ as $\text{\textbf{OutputSensitiveEnum}}(H)$.
\end{itemize}

\begin{algorithm}[h]
\small
\caption{EnumerateMaximalCliquesKtMinorFree}
\label{alg:enum_kt_minor_free}
\KwIn{Graph $G(V, E)$, integer $t$}
\KwOut{$\mathcal{M}(G)$ (All maximal cliques in $G$)}
\SetKwProg{Fn}{Procedure}{}{}

\Fn{ProcEnum($C, P, X$)}{
    \tcc{Step 1: Base Case Termination}
    \If{$P = \emptyset$}{
        \lIf{$X = \emptyset$}{\textbf{Output} $C$ as a maximal clique}
        \Return\;
    }

    \tcc{Step 2: Peeling Vertex Selection (Break ties by smallest degeneracy index)}
    $v^* \leftarrow \arg\min_{v \in P} \deg_{G[P]}(v)$\;
    $d^* \leftarrow \deg_{G[P]}(v^*)$\;

    \tcc{Step 3: Pruning via Small Remaining Graph}
    \If{$|P| \le 1.5t$}{
        $\mathcal{K} \leftarrow \text{OptimizedBK}(G[P])$ 
        
        \For{each $K \in \mathcal{K}$}{
            \lIf{$\forall x \in X, \exists v \in K \text{ s.t. } \{x, v\} \notin E$}{\textbf{Output} $C \cup K$}
        }
        \Return\;
    }

    \tcc{Step 4: Pruning via Highly Dense Terminal Graph}
    \If{$d^* \ge |P| - 0.15\sqrt{|P|}$}{
        $\mathcal{K} \leftarrow \text{OutputSensitiveEnum}(G[P])$ 
        
        \For{each $K \in \mathcal{K}$}{
            \lIf{$\forall x \in X, \exists v \in K \text{ s.t. } \{x, v\} \notin E$}{\textbf{Output} $C \cup K$}
        }
        \Return\;
    }

    \tcc{Step 5: Peeling Process Branching}
    ProcEnum($C \cup \{v^*\}, P \cap N(v^*), X \cap N(v^*)$) \tcc*[r]{Branch A: Include $v^*$}

    ProcEnum($C, P \setminus \{v^*\}, X \cup \{v^*\}$) \tcc*[r]{Branch B: Exclude $v^*$}
}
\BlankLine
\Fn{Main()}{
        Compute a degeneracy ordering $v_1, v_2, \dots, v_n$ of $G$\;
        \For{$i \leftarrow 1$ \KwTo $n$}{
            $P \leftarrow N(v_i) \cap \{v_{i+1}, \dots, v_n\}$\;
            $X \leftarrow N(v_i) \cap \{v_1, \dots, v_{i-1}\}$\;
            ProcEnum($\{v_i\}, P, X$)\;
        }
    }
\end{algorithm}

Let $v_1, \dots, v_n$ be a degeneracy ordering of $G$, and write the forward and backward neighborhoods as:
\begin{align*}
    N^+(v_i) = \{v_j : j > i, v_iv_j \in E(G)\}, \quad N^-(v_i) = \{v_j : j < i, v_iv_j \in E(G)\}.
\end{align*}

\textbf{Correctness.} We first record the invariant maintained by the recursive call \textbf{ProcEnum}($C,$ $ P, X$) initiated at root $v_i$. The set $C$ is a clique containing $v_i$. The set $P$ consists exactly of the not-yet-decided vertices in $N^+(v_i)$ that are adjacent to every vertex of $C$. The set $X$ consists exactly of the already excluded vertices (combining the initial backward neighborhood $N^-(v_i)$ and dynamically excluded vertices) that are still adjacent to every vertex of $C$. The invariant is immediate at the root. When branching on $v^*$, its inclusion restricts both candidates and excluded vertices to $N(v^*)$, whereas its exclusion simply transitions it from $P$ to $X$.

Let $M$ be any maximal clique of $G$, and let $v_i$ be the first vertex of $M$ in the degeneracy ordering. Then $M = \{v_i\} \cup S$ for some $S \subseteq N^+(v_i)$. In the recursion rooted at $v_i$, follow the unique path that includes a branching vertex precisely when it belongs to $S$ and excludes it otherwise. When this path reaches a terminal state $(C, P, X)$ that invokes one of the exact subroutines, the residual set $M \setminus C$ is a maximal clique of $G[P]$; otherwise it could be enlarged by a vertex of $P$, contradicting the maximality of $M$. If the path reaches the base case $P = \emptyset$, then the residual set is empty. 
Hence the corresponding local clique $K = M \setminus C$ is generated by the terminal subroutine whenever a subroutine is called. Since $M$ is globally maximal, no vertex in $X$ is adjacent to all vertices of $C \cup K$, so the explicit output test accepts $M$.

Conversely, suppose the algorithm outputs $C \cup K$. The invariant implies that this set is a clique. It cannot be extended by a vertex of $P$ because $K$ is maximal in $G[P]$ (and in the base case $P = \emptyset$). It cannot be extended by a vertex of $X$ due to the explicit output test ($\forall x \in X, \exists v \in K \text{ s.t. } \{x,v\} \notin E$). Finally, any vertex not in $C \cup P \cup X$ is either not adjacent to the root $v_i$, or failed adjacency to a previously selected vertex of $C$, and hence cannot extend the clique. Thus, every output is a globally maximal clique. The choice of the first vertex $v_i$ in the degeneracy order also implies uniqueness: no maximal clique can be accepted from two different roots.

\textbf{Time Complexity.} To compute the overall time complexity of Algorithm \ref{alg:enum_kt_minor_free}, we must first establish the maximum size of the recursion tree. Our analysis relies fundamentally on the sparsity of $K_t$-minor-free graphs. By Thomason's result \cite{T01}, every $K_t$-minor-free graph has degeneracy
\begin{align*}
     d = O(t\sqrt{\log t}). 
\end{align*}
Consequently, the initial candidate set is bounded by $|N^+(v_i)| \le d$ for every root $v_i$. We use this critical property to bound the number of terminal subproblems generated before the algorithm delegates to the exact subroutines.
\begin{lemma}\label{lem:recursion_depth}
For each fixed root $v_i$, Algorithm \ref{alg:enum_kt_minor_free} generates at most $2^{O(t^{1/2}\log_2^{5/4} t)}$ terminal subproblems before calling one of the two exact subroutines.
\end{lemma}
\begin{proof}
Consider any root-to-leaf path before a terminal condition is triggered, and let $p_j$ denote the size of $P$ immediately before the $j$-th inclusion branch on this path. Since $P$ is always contained in $N^+(v_i)$, we have $p_1 \le d = O(t\sqrt{\log t})$. Whenever the dense terminal condition fails, the chosen minimum-degree vertex $v^*$ satisfies
\begin{align*}
     \deg_{G[P]}(v^*) < |P| - 0.15\sqrt{|P|}. 
\end{align*}
After the inclusion branch, the new candidate set is $P \cap N(v^*)$, ensuring the strict reduction
\begin{align*}
     p_{j+1} < p_j - 0.15\sqrt{p_j}. 
\end{align*}
As long as the small-graph terminal condition is not triggered triggered, $p_j > 1.5t$. To bound the total number of inclusion branches $r$, we partition the range $(1.5t, d]$ into dyadic intervals $(d_{i+1}, d_i]$ where $d_i = d/2^i$. The number of steps $p_j$ spends descending through the $i$-th interval is at most $1 + \frac{d_i - d_{i+1}}{0.15\sqrt{d_{i+1}}} = 1 + \frac{20}{3}\left(\frac{d}{2^{i+1}}\right)^{1/2}$. Summing across all intervals and utilizing the convergent series $\sum_{i=0}^\infty 2^{-(i+1)/2} < 2.5$, the total inclusion depth $r$ is bounded by
\begin{align*}
    r &\le \sum_{i=0}^{\lfloor \log_2 \frac{d}{1.5t} \rfloor} \left( 1 + \frac{20}{3}\left(\frac{d}{2^{i+1}}\right)^{1/2} \right) < {\lceil \log_2 \frac{d}{1.5t} \rceil} + \frac{50}{3}d^{1/2} = O(t^{1/2} \log^{1/4} t) := r_0. 
\end{align*}

We now account for the exclusion branches. By the tie-breaking rule specified in Step 2, we break any minimum-degree ties by selecting the vertex with the smallest index in the degeneracy ordering. Since this tie-breaking rule is deterministic, any root-to-leaf path is uniquely identified by its sequence of included vertices: starting from the root, simulate the recursion by taking exclusion branches until the next vertex in the given sequence is selected as the minimum-degree vertex, take the inclusion branch, and repeat; once the sequence is exhausted, take only exclusion branches until termination. Consequently, every distinct terminal node must originate from a unique inclusion sequence. Since all included vertices are drawn from $N^+(v_i)$, the total number of terminal nodes generated from the root $v_i$ is at most
\begin{equation*}
    \sum_{\ell=0}^{r_0} |N^+(v_i)|^\ell \le d^{r_0+1} = 2^{O(t^{1/2} \log^{5/4} t)}. \qedhere
\end{equation*}
\end{proof}

With the recursion tree bounded by Lemma \ref{lem:recursion_depth}, we now evaluate the computational cost at the leaves. At a small terminal subproblem (Step 3), we have $|P| \le 1.5t$, so \textbf{OptimizedBK} costs at most
\begin{align*}
     O(3^{1.5t/3}) = O(3^{t/2}). 
\end{align*}
At a highly dense terminal subproblem, the complement of $G[P]$ has maximum degree $\Delta \le 0.15\sqrt{|P|}$. Since the small terminal condition did not apply, $|P| > 1.5t$, while $\omega(G[P]) < t$. For sufficiently large $t$,
\begin{align*}
    |P| - \omega(G[P]) > |P| - t > \frac{|P|}{3} \ge 2(0.15\sqrt{|P|})^2 + 2 \ge 2\Delta^2 + 2.
\end{align*}
Thus Lemma \ref{Lem:FW} applies, and since $G[P]$ has no $K_t$-minor,
\begin{align*}
     \left\lfloor\frac{|P| + \omega(G[P])}{2}\right\rfloor = h(G[P]) \le t - 1. 
\end{align*}
Consequently $|P| + \omega(G[P]) < 2t$. Lemma \ref{Lem:vtxCli} gives
\begin{align*}
     \text{mc}(G[P]) \le 4^{2t/5}, 
\end{align*}
and since $|P| \le 2t$, \textbf{OutputSensitiveEnum} running in $O(|V||E|\cdot \text{mc}(G))$ costs at most $\text{poly}(t)4^{2t/5}$ on this leaf.

For a fixed root $v_i$, the total time spent generating candidate outputs across all terminal subroutines is at most $4^{2t/5+o(t)}$, since the local enumeration cost $4^{2t/5}$ dominates $3^{t/2}$ and Lemma \ref{lem:recursion_depth} contributes only $2^{o(t)}$ terminal nodes. Summing over all $n$ roots in the degeneracy ordering, the global generation phase of the algorithm takes at most $n \cdot 4^{2t/5+o(t)}$ time.

We now verify that the local output tests against the exclusion set $X$ preserve this claimed linear dependence on $n$. Validating each candidate requires checking its adjacency against all vertices in $X$. Since the maximum clique size is bounded by $t$, testing a candidate against a single vertex in $X$ takes $O(t)$ time. This polynomial overhead is safely absorbed into the $2^{o(t)}$ factor. The exclusion set $X$ consists of two parts: the initial backward neighborhood $N^-(v_i)$, and the vertices dynamically excluded from $P$.

Because the dynamically excluded vertices are drawn from $N^+(v_i)$, their number is bounded by the degeneracy $d = O(t\sqrt{\log t}) $. Checking this dynamic part against the outputs of $v_i$ takes $O(t\sqrt{\log t}) \cdot 4^{2t/5+o(t)} = 4^{2t/5+o(t)}$ time. Across all $n$ roots, this specific cost contributes at most $n \cdot 4^{2t/5+o(t)}$.

For the static part, checking all candidate outputs against $N^-(v_i)$ over all roots costs explicitly
\begin{align*}
    \sum_{i=1}^n |N^-(v_i)| 4^{2t/5+o(t)} = |E(G)| 4^{2t/5+o(t)}. 
\end{align*}
As $G$ is $d$-degenerate, we have $|E(G)| \le dn = n\cdot O(t\sqrt{\log t})$, meaning this static checking cost is also bounded by $n \cdot 4^{2t/5+o(t)}$. 

The total running time of the algorithm is the sum of the time spent generating the candidate outputs and the time spent validating them. Since the generation phase across all $n$ roots takes $n \cdot 4^{2t/5+o(t)}$ time, and both validation costs (against the dynamic and static parts of $X$) are asymptotically bounded by the same term, the overall running time is bounded by
\begin{align*}
     n \cdot 4^{2t/5+o(t)}. 
\end{align*}

\textbf{Asymptotic Optimality in the exponential term.} Consider an $n$-vertex graph formed
by taking disjoint copies of the complete multipartite graph $K_{4,4,\dots,4}$ with $\left\lfloor (2t-1)/5 \right\rfloor$ parts,
plus isolated vertices if needed. Each component has $4\left\lfloor (2t-1)/5 \right\rfloor$ vertices and exactly
$4^{\left\lfloor (2t-1)/5 \right\rfloor}$ maximal cliques. By Lemma \ref{Lem:FW}, the Hadwiger number of each component is at
most $t-1$, so the graph is $K_t$-minor-free. Hence the output size is
\begin{align*}
     n \cdot 4^{2t/5 - o(t)}.
\end{align*}
Thus the upper bound is optimal in its exponential dependence on $t$, up to the $o(t)$ term in the exponent.
%


\section{Maximal cliques in $K_t$-immersion-free graphs}\label{sec:immersion_alg}

In this section, we prove  Theorem \ref{Thm:immersion} by analyzing a customized algorithm (\textbf{Algorithm \ref{alg:enum_kt_immersion_free}}) tailored for $K_t$-immersion-free graphs. Unlike the framework for minor-free graphs, Algorithm \ref{alg:enum_kt_immersion_free} omits the density-driven pruning step. Instead, it relies on tracking the localized minimum degree to dynamically shrink the candidate set until it falls below a strict size threshold ($|P| \le t$), at which point it delegates the residual graph to the optimal exact subroutine.
\begin{algorithm}[h]
\caption{EnumerateMaximalCliquesKtImmersionFree}
\label{alg:enum_kt_immersion_free}
\KwIn{Graph $G(V, E)$, integer $t$}
\KwOut{$\mathcal{M}(G)$ (All maximal cliques in $G$)}
\SetKwProg{Fn}{Procedure}{}{}

\Fn{ProcEnum($C, P, X$)}{
    \tcc{Step 1: Base Case Termination}
    \If{$P = \emptyset$}{
        \lIf{$X = \emptyset$}{\textbf{Output} $C$ as a maximal clique}
        \Return\;
    }

    \tcc{Step 2: Peeling Vertex Selection (Break ties by smallest degeneracy index)}
    $v^* \leftarrow \arg\min_{v \in P} \deg_{G[P]}(v)$\;
    $d^* \leftarrow \deg_{G[P]}(v^*)$\;

    \tcc{Step 3: Pruning via Small Remaining Graph}
    \If{$|P| \le t$}{
        $\mathcal{K} \leftarrow \text{OptimizedBK}(G[P])$ 
        
        \For{each $K \in \mathcal{K}$}{
            \lIf{$\forall x \in X, \exists v \in K \text{ s.t. } \{x, v\} \notin E$}{\textbf{Output} $C \cup K$}
        }
        \Return\;
    }

    \tcc{Step 4: Peeling Process Branching}
    ProcEnum($C \cup \{v^*\}, P \cap N(v^*), X \cap N(v^*)$) \tcc*[r]{Branch A: Include $v^*$}

    ProcEnum($C, P \setminus \{v^*\}, X \cup \{v^*\}$) \tcc*[r]{Branch B: Exclude $v^*$}
}
\BlankLine
\Fn{Main()}{
        Compute a degeneracy ordering $v_1, v_2, \dots, v_n$ of $G$\;
        \For{$i \leftarrow 1$ \KwTo $n$}{
            $P \leftarrow N(v_i) \cap \{v_{i+1}, \dots, v_n\}$\;
            $X \leftarrow N(v_i) \cap \{v_1, \dots, v_{i-1}\}$\;
            ProcEnum($\{v_i\}, P, X$)\;
        }
    }
\end{algorithm}

\textbf{Correctness.} The correctness proof is identical to that of Algorithm \ref{alg:enum_kt_minor_free}. Every global maximal clique is uniquely assigned to its first vertex in the fixed degeneracy ordering. The local recursion decides which forward neighbors are included, and the explicit validity test against the exclusion set $X$ at the leaf nodes mirrors the condition that no earlier neighbor (from the backward neighborhood $N^-(v_i)$) and no explicitly excluded later neighbor (from the forward neighborhood $N^+(v_i)$) can extend the clique.

\textbf{Time Complexity.} To compute the overall time complexity of Algorithm \ref{alg:enum_kt_immersion_free}, we must first bound the number of terminal subproblems generated from any fixed root $v_i$ before the recursion terminates (Step 1) or delegates to the exact subroutine (Step 3). To achieve this, we introduce a potential function $\Phi(P) = \max(0, |P| - t)$ that quantifies the ``excess size'' of the candidate set relative to the pruning threshold, allowing us to mathematically bound the maximum depth of the inclusion branches.
\begin{lemma}\label{Lem:Alg2}
    For each fixed root $v_i$, Algorithm \ref{alg:enum_kt_immersion_free} generates at most $2^{O(\log^2 t)}$ terminal subproblems before calling the exact subroutine.
\end{lemma}
\begin{proof}
Let $P^{(j)}$ be the candidate set after $j$ inclusion branches from root $v_i$, with $P^{(0)} = N^+(v_i)$. Since $K_t$-immersion-free graphs are $(11t+6)$-degenerate \cite{DDZF+14,DY18}, $|P^{(0)}| \le 11t+6$.
Prior to pruning ($|P^{(j-1)}| > t$), Branch A sets $P^{(j)} = P^{(j-1)} \cap N(v^*)$, where $v^* = \arg\min_{v \in P^{(j-1)}} \deg_{G[P^{(j-1)}]}(v)$. 
Since we are prior to pruning, $|P^{(j-1)}| > t$.
If the minimum degree of $G[P^{(j-1)}]$ strictly exceeded $\frac{|P^{(j-1)}|+t}{2} - 2$, we could arbitrarily select a subset $T \subseteq P^{(j-1)}$ of size $t$. Every vertex in $T$ would satisfy the high-degree condition of Lemma \ref{Lem:FW20}, forcing a $K_t$-immersion and yielding a contradiction.
Hence, 
\begin{align*}
    \deg_{G[P^{(j-1)}]}(v^*) \le \frac{|P^{(j-1)}|+t}{2} - 2.
\end{align*}
Thus, 
\begin{align*}
    \Phi(P^{(j)}) \le \max\{0, \frac{|P^{(j-1)}| - t}{2} - 2\} < \frac{1}{2}\Phi(P^{(j-1)}).
\end{align*}
Given the initial excess $\Phi(P^{(0)}) \le 10t+6$, the pruning condition $\Phi(P) = 0$ is universally triggered after at most $r_0 = \lceil \log_2(10t+6) \rceil$ inclusion branches.

We now account for the exclusion branches (Branch B). By the tie-breaking rule specified in Step 2, we break any minimum-degree ties by selecting the vertex with the smallest index in the degeneracy ordering. Since this choice is deterministic, one can reconstruct any root-to-leaf path solely from its ordered sequence of included vertices: starting from the root state $P^{(0)}$, simulate the recursion by taking exclusion branches until the next vertex in the given sequence is selected as the minimum-degree vertex, take the inclusion branch, and repeat; once the sequence is exhausted, take only exclusion branches until termination. 

Consequently, every distinct terminal node must originate from a unique inclusion sequence. Since all included vertices are drawn from $N^+(v_i)$, and any such sequence contains at most $r_0$ vertices, the total number of terminal nodes generated from the root $v_i$ is bounded by
\begin{equation*}
    \sum_{\ell=0}^{r_0} |N^+(v_i)|^\ell < (11t+6)^{r_0+1} = 2^{O(\log^2 t)}. \qedhere
\end{equation*}
\end{proof}
With the recursion tree bounded by Lemma \ref{Lem:Alg2}, we evaluate the computational cost at the leaves. The generated leaves exclusively correspond to the two termination conditions in Algorithm \ref{alg:enum_kt_immersion_free}: the base case ($P = \emptyset$) and the small-graph pruning ($|P| \le t$). 
At a base case leaf, excluding the later validation against the exclusion set $X$, the computational cost is bounded by $O(t)$ since $\omega(G) < t$ in a $K_t$-immersion-free graph. 
At a pruning leaf, the exact subroutine \textbf{OptimizedBK} enumerates all local maximal cliques in $O(3^{|P|/3}) \le O(3^{t/3})$ time. Since the exponential term $3^{t/3}$ asymptotically dominates $O(t)$, the maximum enumeration cost at any terminal subproblem is bounded by $O(3^{t/3})$.

For a fixed root $v_i$, the total time spent generating candidate outputs across all terminal subroutines is at most $3^{t/3+o(t)}$, since the local enumeration cost $O(3^{t/3})$ dominates and Lemma \ref{Lem:Alg2} contributes only $2^{O(\log^2 t)} = 2^{o(t)}$ terminal nodes. Summing over all $n$ roots in the degeneracy ordering, the global generation phase of the algorithm takes at most $n \cdot 3^{t/3+o(t)}$ time.

We now verify that the local output tests against the exclusion set $X$ preserve this claimed linear dependence on $n$. Validating each candidate requires checking its adjacency against all vertices in $X$. Since the size of any valid clique in a $K_t$-immersion-free graph is bounded by $t$, testing a candidate against a single vertex in $X$ takes $O(t)$ time. This polynomial overhead is safely absorbed into the $2^{o(t)}$ factor. The exclusion set $X$ consists of two parts: the initial backward neighborhood $N^-(v_i)$, and the vertices dynamically excluded from $P$.

Because the dynamically excluded vertices are drawn from $N^+(v_i)$, their number is bounded by $11t+6$. Checking this dynamic part against the outputs of $v_i$ takes at most $(11t+6) \cdot 3^{t/3+o(t)}$ time. Since any polynomial factor is safely absorbed into the $3^{o(t)}$ asymptotic factor, this bound simplifies to $3^{t/3+o(t)}$. Across all $n$ roots, this specific validation cost contributes at most $n \cdot 3^{t/3+o(t)}$.

For the static part, checking all candidate outputs against $N^-(v_i)$ over all roots costs explicitly
\begin{align*}
     \sum_{i=1}^n |N^-(v_i)| 3^{t/3+o(t)} = |E(G)| 3^{t/3+o(t)}. 
\end{align*}
As $G$ is $(11t+6)$-degenerate, we have $|E(G)| \le (11t+6)n = n 2^{o(t)}$, meaning this static checking cost is also bounded by $n \cdot 3^{t/3+o(t)}$.

The total running time of the algorithm is the sum of the time spent generating the candidate outputs and the time spent validating them. Since the generation phase across all $n$ roots takes $n \cdot 3^{t/3+o(t)}$ time, and both validation costs (against the dynamic and static parts of $X$) are asymptotically bounded by the same term, the overall running time is bounded by
\begin{align*}
     n \cdot 3^{t/3+o(t)}. 
\end{align*}

\textbf{Asymptotic Optimality in the exponential term.} Let $G$ be the disjoint union of
$\left\lfloor n/(3\left\lfloor (t-1)/3 \right\rfloor) \right\rfloor$ copies of the complete multipartite graph $K_{3,3,\dots,3}$ with $\left\lfloor (t-1)/3 \right\rfloor$ parts,
together with isolated vertices if necessary. Each component has at most $t-1$ vertices and
is therefore $K_t$-immersion-free. It contributes exactly $3^{\left\lfloor (t-1)/3 \right\rfloor}$ maximal cliques. Thus the
total output size is
\begin{align*}
n \cdot 3^{t/3 - o(t)}.
\end{align*}
The upper bound is therefore optimal in the exponential dependence on $t$, up to the $o(t)$
term in the exponent.
\section{Concluding Remarks}
In this paper, we resolve the exact parameter dependence for the maximal clique enumeration problem within sparse graphs characterized by excluded minors and immersions. By leveraging a dynamic peeling strategy and the connections between local density and global structural exclusion, we significantly tighten existing complexity bounds. We demonstrate that enumerating all maximal cliques can be achieved in $n \cdot 4^{2t/5+o(t)}$ time for $K_t$-minor-free graphs and $n \cdot 3^{t/3+o(t)}$ time for $K_t$-immersion-free graphs, and these bounds are asymptotically optimal up to sub-exponential factors.

Beyond algorithmic improvements, our results highlight a fundamental theoretical insight: the underlying mechanism of a structural restriction determines the exact computational boundary of the enumeration task. Specifically, forbidding a $K_t$-minor bounds the graph's density via vertex contractions, directly yielding the optimal exponential base of $4^{2/5}$. In contrast, forbidding a $K_t$-immersion restricts the graph via edge-disjoint routing, leading to a different density limit that shifts the optimal base to $3^{1/3}$.

While our bounds firmly settle the optimal exponential base for parameter $t$ in these sparse frameworks, our work opens several compelling avenues for future research. A natural progression is to investigate whether similar tight algorithmic bounds can be derived for graphs excluding other structures, such as topological minors or specific subdivisions. Furthermore, while our framework is primarily theoretical, it is a highly promising direction to adapt our dynamic peeling techniques to practical community detection algorithms on massive real-world networks.
\section{Acknowledgments}
Research supported by National Key R\&D Program of China (Grant No. 2023YFA1010202) and by the Central Guidance on Local Science and Technology Development Fund of Fujian Province (Grant No.~2023L3003). 
\section*{Competing Interest}
The authors declare that they have no competing interests or financial conflicts to disclose.

\bibliographystyle{abbrv}
\bibliography{maximalClique}
\end{document}